\documentclass[10pt,conference]{IEEEtran}

\IEEEoverridecommandlockouts

\usepackage{amsmath,amssymb,amsfonts}%
\usepackage{amsthm}%
\usepackage{mathrsfs}%
\usepackage{bm}

\usepackage{cite}

\usepackage{url}
\usepackage{graphicx}
\usepackage{epstopdf}
\usepackage{xspace}
\usepackage{subfigure}

\usepackage{multirow}
\usepackage{booktabs} 
\usepackage{siunitx}
\usepackage{makecell}
\usepackage{makecell}

\usepackage{xcolor}
\usepackage{todonotes}
\usepackage{algorithm}
\usepackage{algpseudocode}
\usepackage{balance}
\usepackage{caption} 
\usepackage[switch]{lineno} 

\newtheorem{theorem}{Theorem}

\theoremstyle{definition}
\newtheorem{definition}{Definition}[section]
\newtheorem*{remark}{Remark}

\usepackage{comment}
\usepackage[inline]{enumitem}
\usepackage[normalem]{ulem}
\usepackage{colortbl}

\begin{document}


\title{FedBlock: A Blockchain Approach to Federated Learning against Backdoor Attacks}

\author{
\IEEEauthorblockN{Duong H. Nguyen\IEEEauthorrefmark{1}, Phi L. Nguyen\IEEEauthorrefmark{1},
Truong T. Nguyen\IEEEauthorrefmark{2},
Hieu H. Pham\IEEEauthorrefmark{3},
Duc A. Tran\IEEEauthorrefmark{4}}
\IEEEauthorrefmark{1} Hanoi University of Science and Technology, Hanoi, Vietnam\\
\IEEEauthorrefmark{2} National Institute of Advanced Industrial Science and Technology, Tokyo, Japan\\
\IEEEauthorrefmark{3}  VinUniversity, Hanoi, Vietnam\\
\IEEEauthorrefmark{4}  University of Massachusetts, Boston, USA  (corresponding author)\\
{\small
Email: {duong.nh190044, lenp}@soict.hust.edu.vn, nguyen.truong@aist.go.jp,
 hieu.ph@vinuni.edu.vn, duc.tran@umb.edu
 }
}

\maketitle
\thispagestyle{plain}
\pagestyle{plain}

\begin{abstract}
Federated Learning (FL) is a machine learning method for training with private data locally stored in distributed machines without gathering them into one place for central learning. Despite its promises, FL is prone to critical security risks. First, because FL  depends on a central server to aggregate local training models, this is a single point of failure. The server might function maliciously. Second, due to its distributed nature, FL might encounter backdoor attacks by participating clients. They can poison the local model before submitting to the server. Either type of attack, on the server or the client side, would severely degrade learning accuracy. We propose FedBlock, a novel blockchain-based FL framework that addresses both of these security risks. FedBlock is uniquely desirable in that it involves only smart contract programming, thus deployable atop any blockchain network. Our framework is  substantiated with a comprehensive evaluation study using real-world datasets. Its robustness against backdoor attacks is competitive with the literature of FL backdoor defense. The latter, however, does not address the server risk as we do.
\end{abstract}


\section{Introduction}\label{sec:intro}

Federated Learning (FL) \cite{pmlr-v54-mcmahan17a}  is a Machine Learning approach to learning  a model using distributed training data that remain private and unmoved on local machines. A typical FL architecture consists of these local machines, hereafter referred to as the ``clients", which have the training data,  and a central server to coordinate the training, called the ``aggregation" server. The learning is an iterative procedure. In the first step, the aggregation server broadcasts a global learning model, initially random, to all   clients. In the second step,  each client in parallel performs local training on its own local data to improve this model. In the third step, the clients  send their respective improved models to the server who in turn aggregates  them to obtain a new global model. Then the first step is repeated until the global model converges.

FL is elegant in idea and has widespread applications \cite{MAL-083}. However, its success relies on the server working normally and clients being good citizens. This is not always true in practice. Indeed, FL is vulnerable to two types of security risks: server attack and backdoor attack.

\underline{Server attack}: Like in any client/server system, the centralized server is a security bottleneck. In FL, an attacker who gains control over the aggregation server can distribute bad global models to the clients, thus poisoning the entire system. Naturally, besides typical solutions designed to secure a server, the best way to avoid server attacks is by deploying multiple  servers to decentralize this single point of failure.  However, we then face a new challenge of how to coordinate these servers because the ``coordinator" could itself become a new single point of failure. 

 \underline{Backdoor attack}: Dishonest clients can send fake local models to the server.  In ``untargeted" backdoor attacks, the poisoned  global model leads to randomly bad predictions. In ``targeted" backdoor attacks, the attacker's goal  is to make the global model always predict according to a targeted outcome. Backdoor attacks are common threats to computer systems. According to   the IBM Security X-Force Threat Intelligence Index 2023, they are the top action by cybercriminals. Nearly a quarter of cyber incidents last year involved backdoor attacks. Recently, backdoor danger has been demonstrated with FL systems, raising serious concerns \cite{wang2020attack}.
 
In this paper, we are interested in immunizing FL against the above attacks. Our approach is inspired by blockchain technology \cite{TranThaiKrishnamachari22}. Blockchain is a computing solution for executing transactions in a way that is honest, immutable, and traceable. Because blockchain runs on a decentralized network of many autonomous computers with majority-based consensus, it is unaffected by any attack. Our idea is that, instead of running the centralized model aggregation on a server, we run it on the ``blockchain computer". Precisely, we implement the task of centralized aggregation as a smart contract to deploy on a blockchain network. An immediate benefit is that the risk due to server attacks becomes non-existent. Our goal then is narrowed to how to mitigate backdoor attacks. 

From the system-wide perspective, our work is a fresh direction compared to the literature. FL defense against server attacks and backdoor attacks   mainly serves non-blockchain settings \cite{nguyen2023backdoor}. Meanwhile,  blockchain use for FL has started only recently with purposes mostly not about backdoor risks, for example, to remove the dependency on the central server \cite{Madill2022ScaleSFLAS}, incentivize clients to contribute training data \cite{9293091}, or provide robustness to non-backdoor security vulnerabilities \cite{10.1145/3579654.3579700}. Existing research on backdoor defense in blockchain-based FL is rare \cite{electronics12112500}, which requires creating a dedicated blockchain network, a difficult task with limited adoption. 

In contrast, we make the following key contributions:
\begin{itemize}
\item We propose FedBlock, a blockchain-based FL framework that provides simultaneous defense against server attacks and backdoor attacks, works at the smart contract level, and hence can run on any blockchain network. To our knowledge, this is the first framework with such features.
\item We propose an implementation of FedBlock integrating a realistic backdoor-defense technique. The result is a novel decentralized backdoor defense for FL. It uses an efficient amount of model data that is discriminative enough to distinguish compromised clients from the benign. 
\item We justify and validate our technique with in-depth experiments. We evaluate comprehensive scenarios using various datasets, threat models, and attack scenarios. The results demonstrate FedBlock's viability and superiority compared to the literature. 
\end{itemize}

The remainder of the paper is organized as follows. Related work is reviewed in Section \ref{sec:related}. Some preliminaries on federated learning and backdoor models are provided in  Section \ref{sec:background}. The FedBlock framework  is proposed in Section \ref{sec:solution}, followed by a specific backdoor-defense implementation in Section \ref{sec:fedblockimplement}. The evaluation results  are discussed in Section \ref{sec:eval}.  The paper concludes in Section \ref{sec:conclusion} with pointers to our future work.

\section{Related Work}\label{sec:related}
Backdoor research in machine learning (ML) is increasing exponentially in the literature \cite{nguyen2023backdoor}. In particular, backdoor attacks  are serious concerns for FL due to the possibility of stealthily injecting a malevolent behavior in local models to affect the global model \cite{bagdasaryan2020backdoor, wang2020attack}. Unlike ML backdoors, an adversary in FL can insert poison in various iterations of the FL training pipeline, making backdoors more challenging to counter. 
Their impact is well-aware in many FL-applicable scenarios such as CV \cite{bagdasaryan2020backdoor},  NLP \cite{wang2020attack}, and healthcare \cite{jin2022backdoor}.

\textbf{FL backdoor defense. }
There is growing traction on defense against FL backdoors \cite{Shen2016AurorDA, Sun2019CanYR, wu2022federated, Wu2020MitigatingBA}. One can counter  backdoor attacks in different phases of the learning process: pre-aggregation, in-aggregation, and post-aggregation. Pre-aggregation defense \cite{Shen2016AurorDA, nguyen2022flame} aims to identify and mitigate the impact of malicious updates before  global update. In-aggregation defense \cite{Sun2019CanYR, pillutla2022robust} uses new aggregation operators customized to fight backdoors. Post-aggregation defense  \cite{wu2022federated, Wu2020MitigatingBA} provides repair to the model damaged after completing the FL training process.

For example, Xie et al. \cite{xie2021crfl} introduced a defense mechanism to cleanse local updates by clipping and perturbation. Backdoor-robust aggregation operators can be $\alpha$-trimmed mean (Yin et al., \cite{Yin2018ByzantineRobustDL}) or geometric median (Pillutla et al., \cite{pillutla2022robust}). Clustering is a favorite method to detect and remove poisoned models; e.g., k-mean used in Shen et al. \cite{Shen2016AurorDA}, cosine similarity in  Sattler et al. \cite{Sattler2020OnTB}.  FLARE \cite{10.1145/3488932.3517395} leverages the penultimate layer representations of the model for characterizing the adversarial influence on each local model update.

In other works, Wu et al. \cite{Wu2022TowardCB} focus on the activation of ``backdoor neurons" within neural networks when exposed to backdoor images and propose two strategies to identify and prune these neurons. 
FLIP \cite{zhang2023flip} requires benign clients to train the local model on crafted backdoor triggers capable of inducing misclassification, which counters  data poisoning by the attacker. Upon global aggregation, the injected backdoor features present in the aggregated  model undergo a mitigation process by applying to the benign clients. FLDetector \cite{Zhang2022FLDetectorDF} is based on observing the inconsistency of a malicious client's model updates in different FL iterations. The server predicts each client’s model update in each iteration based on historical model updates and compare with the received model to detect malice. BayBFed \cite{10179362} computes a probabilistic measure over the clients’ updates to   track  adjustments and uses a novel detection algorithm leveraging this measure to detect and remove malice.
To date, all existing countermeasures serve centralized FL and assume that the aggregation server is honest. Readers are referred to \cite{nguyen2023backdoor} for a recent comprehensive survey on FL backdoor attack and defense.

\textbf{Blockchain-based FL.} Blockchain brings benefit to FL in areas regarding participation incentive, data integrity and server vulnerability. Recently, as blockchain goes mainstream, efforts on leveraging this technology as an alternative to cloud-based FL are on the rise. 
Li et al. \cite{9293091} proposed a blockchain-based FL framework with an incentive mechanism based on Committee Consensus. This mechanism mitigates the computational overhead associated with consensus decision making. Lin et al. \cite{10.1145/3579654.3579700} extended this work with the new mechanism of Dual Committee Consensus to allow large numbers of nodes to join, which provides better security. To enhance scalability, Madill et al. \cite{Madill2022ScaleSFLAS} applies sharding to blockchain-based FL. Instead of using one blockchain for the entire FL system, clients are assigned to different blockchain shards where shard-level aggregation takes place. The resulted shard models are then aggregated globally. Blockchain is also integrated in FL systems leveraged by edge computing. 

Regarding backdoor attacks in the context of blockchain-based FL, very rare works can be found, and only recently \cite{10001370,electronics12112500}. To our knowledge, the only effort on designing an FL backdoor defense using blockchain is by Li et al. \cite{electronics12112500}. They proposed building a new blockchain network from scratch customized for this specific purpose. Its feasibility is questionable, not to mention its limited adoption. In contrast, our proposed blockchain-based FL framework is blockchain-agnostic. It works at the smart contract level, thus deployable on any blockchain network.

\section{Preliminaries} \label{sec:background}
Consider a  learning task mapping an input object in $X$ to an output label in $Y$, given a  training set of  samples, $\mathcal{D} = \{(x_1, y_1), ..., (x_{|\mathcal{D}|}, y_{|\mathcal{D}|})\}$ where   $x_{i} \in X$ is the feature vector of the $i^{th}$ input sample and $y_i \in Y$  its corresponding label. Typically, the objective is to find a learning model, represented by $\mathbf{w} \in \mathbb{R}^d$, that minimizes the following empirical loss 
\begin{align}
F(\mathbf{w}; \mathcal{D}) \triangleq   \frac{1}{|\mathcal{D}|} \sum_{(x, y) \in \mathcal{D}} l\bigg(\mathbf{w};x, y\bigg)
\end{align}
where  ${l}(\mathbf{w}; x,y)$ is a user-defined function measuring the prediction loss  on sample $(x, y)$  using  model $\mathbf{w}$. 
A common way to find $\mathbf{w}$ is by Stochastic Gradient Descent.
\begin{algorithm}
\caption{Stochastic Gradient Descent (SGD)}
\label{alg:SGD}
\begin{enumerate}
\item  Start with some initial model $\mathbf{w}^{(0)}$
\item For each step $t$ = 1, 2, ..., $\tau$: compute the update below
	\[
	\mathbf{w}^{(t)} = \mathbf{w}^{(t-1)} - \eta \nabla_w F(\mathbf{w}^{(t-1)}; \mathcal{D})
	\]
\item Return $\mathbf{w}^{(\tau)}$
\end{enumerate}
\end{algorithm}

Here, $\nabla$ is the vector differential operator in math. As the number of steps $\tau$ is sufficiently large, the value $\mathbf{w}^{(\tau)}$ at the end of this loop should converge to the optimum $\mathbf{w}$. Parameter $\eta$ is predefined, called the learning rate.  Denote this algorithm by $\texttt{SGD}(\mathbf{w}^{(0)}; \mathcal{D})$ where $\mathbf{w}^{(0)}$ is the initial model to begin the loop with  and $\mathcal{D}$ the set of training samples. 

\subsection{Federated Learning}

In the setting of Federated Learning (FL), the training samples are not available all at one place, but instead they reside independently and privately on many local client machines. Let $N$ be the number of clients and $\mathcal{D}$ = $\mathcal{D}_1  \cup \mathcal{D}_2 \cup  ...\cup \mathcal{D}_N$ where $\mathcal{D}_i$ denotes the  subset of training samples owned by client $i \in [N]$. 
The prediction loss can then be expressed as 
\begin{align*}
F(\mathbf{w}; \mathcal{D}) = \frac{1}{|\mathcal{D}|} \sum_{(x, y) \in \mathcal{D}} l\bigg(\mathbf{w};x, y\bigg) \nonumber\\
= \sum_{i=1}^N \frac{|\mathcal{D}_i|}{|\mathcal{D}|} \bigg(\underbrace{ \frac{1}{|\mathcal{D}_i|} \sum_{(x, y) \in \mathcal{D}_i} l(\mathbf{w}; x, y)}_{F(\mathbf{w}; \mathcal{D}_i)} \bigg)
= \sum_{i=1}^N \frac{|\mathcal{D}_i|}{|\mathcal{D}|} F(\mathbf{w}; \mathcal{D}_i). \label{eq:predictionloss}
\end{align*}
We can think of  $F(\mathbf{w}; \mathcal{D}_i)$ as the local prediction loss of client $i$.  This suggests  a simple distributed learning approach: the clients  can each independently  solve the learning problem  using their own  training data and the average over all these local models provides a good approximation for the optimal model. This is the foundation for FL.  

The most popular FL algorithm arguably is FedAvg \cite{pmlr-v54-mcmahan17a}. FedAvg  uses SGD presented above as the learning method. FedAvg runs repetitively a number of rounds, called ``aggregation rounds" (a.k.a. ``communication" or ``global" rounds in the literature). In each round, the server averages the clients' local models to update the global model. This new global model is then sent to all clients for re-training (in practice, a random subset to save costs).  Basically, it works as follows.
\begin{algorithm}
 \caption{Federated Averaging (FedAvg)}
 \label{alg:FedAvg}
\begin{enumerate}
\item The server starts with an initial model $\mathbf{w}^{(0)}$
\item For each round $t = 1 \rightarrow T$ (called ``aggregation round"):
	\begin{enumerate}
	\item The server broadcasts model $\mathbf{w}^{(t-1)}$ to a random subset $\mathcal{C}^{(t)}$ of $K \le N$  clients
	\item In response, each  client $i \in \mathcal{C}^{(t)}$ runs SGD-based local training to improve its local model 
		\[
		\mathbf{w}^{(t)}_i = \textup{\texttt{SGD}}(\mathbf{w}^{(t-1)}; \mathcal{D}_i)
		\]
		and sends it to the server. This SGD starts with the  global model $\mathbf{w}^{(t-1)}$ just received from the server.
	\item The server updates the global model by averaging the updated local models, weighted by data size:
	\[
	\mathbf{w}^{(t)} = \frac{\sum_{i \in \mathcal{C}^{(t)}} {|\mathcal{D}_i|} \mathbf{w}^{(t)}_i }{ \sum_{i \in \mathcal{C}^{(t)}} {|\mathcal{D}_i|}}
	\]
	\end{enumerate}
\item Return $\mathbf{w}^{(T)}$
\end{enumerate}
\end{algorithm}

This algorithm has several variants \cite{pmlr-v54-mcmahan17a}. For example, different criteria may be considered to choose the subset $\mathcal{C}^{(t)}$ of clients to receive the global model update. Also, the SGD procedure can be modified to apply on a mini-batch  of  samples in each gradient descent step instead of the entire training set.

\subsection{Backdoor Attacks}
Backdoor attacks, or ``backdoors", in machine learning are techniques that implant secret behaviors into trained models. This way, for example, an adversary can create a backdoor that bypasses a face recognition system.
FL is highly vulnerable to backdoors due to the autonomy of FL clients in training their model. In our threat model, 
we assume that a backdoor attacker takes control of a subset of clients. 
The attacker can poison their training data and training model.
\begin{itemize}
   \item \textbf{Data Poisoning:} This is the most popular method. Trigger artifacts are injected into   the training data by ways  such as edge-case backdoor \cite{wang2020attack} and trigger-based backdoor \cite{bagdasaryan2020backdoor, xie2020dba}. The poisoned model works correctly as usual on normal training data, but if the input data is crafted to include the triggering patterns, the model will predict a targeted result that is pre-chosen by the attacker. 
   \item \textbf{Model Poisoning:} To amplify the impact,  the attacker can scale the local gradients of the updated model before submitting it to the server \cite{bagdasaryan2020backdoor}. The attack can go stealthier by using Projected Gradient Descent (PGD) \cite{Sun2019CanYR} to make the poisoned model look like the global model in each aggregation round. PGD can further be combined with model replacement to maximize attack effectiveness \cite{wang2020attack}.
\end{itemize}

In our threat model, the attacker knows the server's aggregation logic and the other uncompromised clients' training logic, but cannot modify any of these processes.

\section{Proposed Framework}\label{sec:solution}
We propose FedBlock, a backdoor-robust FL framework that runs on a blockchain network. The blockchain can be any smart contract network at layer-1 such as Ethereum and Algorand or at layer-2 such as Optimism and Arbitrum. A smart contract blockchain is one that enables development of arbitrary applications as smart contracts. For example, Bitcoin is not a smart contract blockchain because it is designed specifically to serve one application which is payment, whereas Ethereum is a smart contract blockchain because it allows any application to be developed atop.

\subsection{System Overview}
In a nutshell, as illustrated in Figure \ref{fig:fedblock_architecture}, FedBlock has the following stakeholders: 
\begin{itemize}
\item \textbf{The smart contract}: No server is needed. Instead, we  deploy a smart contract on the blockchain to implement the logic of model aggregation that otherwise would run on the server as in cloud-based FL. This aggregation mostly involves number averaging, and so the computation is simple enough to be implemented as a smart contract. Because every smart contract is open-sourced and immutable, its correctness and honesty are verifiable. Server attacks are thus non-existent.
\item \textbf{The clients}: Clients run local training on their own machine, not on the blockchain, as usual. Once a client has finished training, it uploads the updated local model by calling a function of the smart contract. In the other way back, once the smart contract has updated the global model, it triggers a blockchain event. Clients watch this event and when notified call the smart contract to download the updated global model.  
\item \textbf{The verifiers}: Due to poison possibility,  local models need to be verified of quality before being aggregated by the smart contract. No single entity should be trusted to provide this service. Instead, we propose a decentralized  method relying on a set of verfifiers. A client or anyone can be a verifier, who watches the  blockchain for events triggering a verification need. When notified, each verifier performs the verification locally and outputs a score to the smart contract. The smart contract uses these scores to prioritize trusted models over suspected ones.    
\end{itemize}
\begin{figure}
    \centering
    \includegraphics[width=0.45\textwidth]{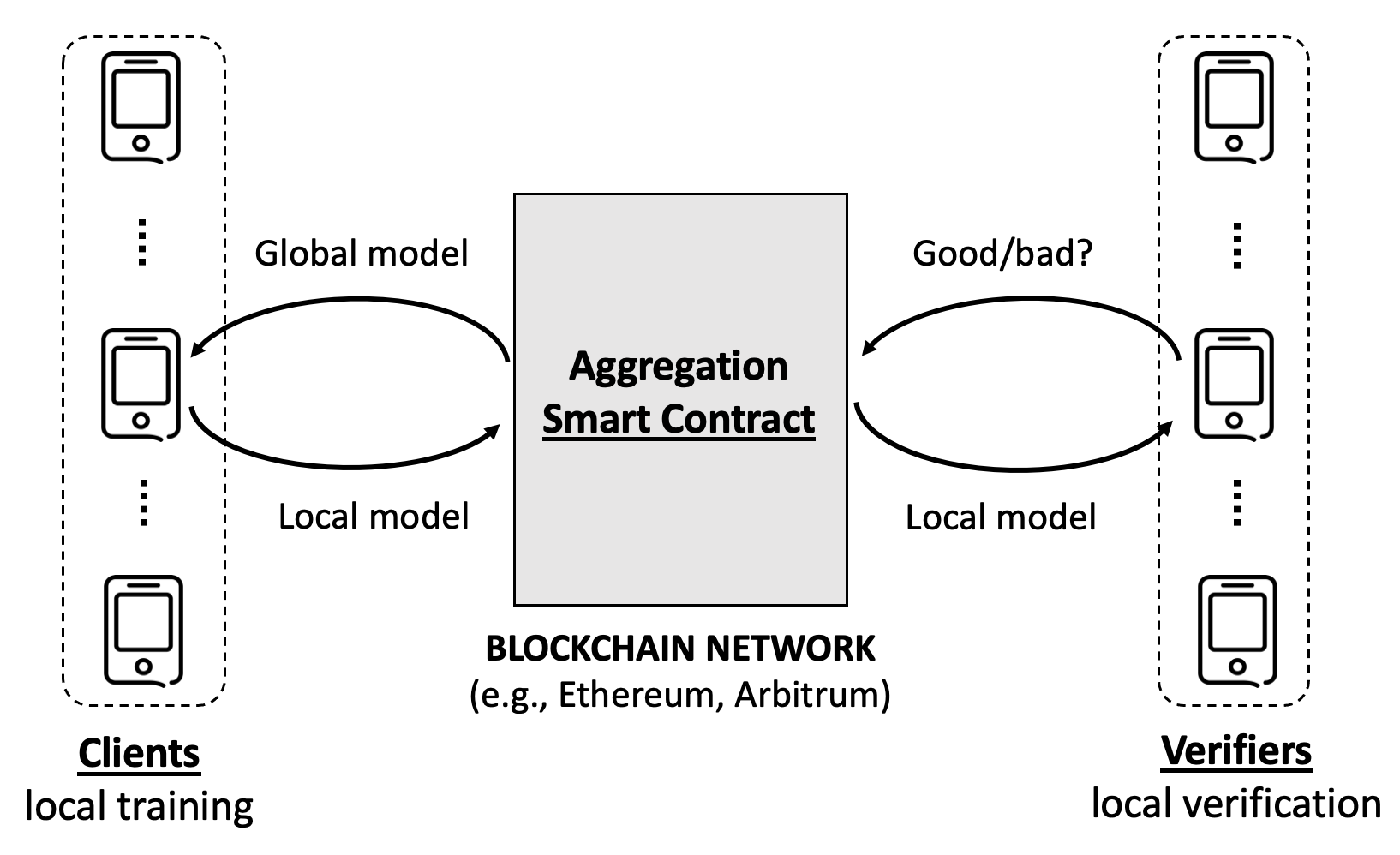}
    \caption{Overview architecture of FedBlock. }
    \label{fig:fedblock_architecture}
\end{figure}

To realize this architecture, we are aware of blockchain drawback. In blockchain, it costs a network fee, called ``gas fee", to call a smart contract function if the call results in changing smart contract data. As the size of a training model can be substantial, given repeated uploads of local models, the total gas fee is a concern. Secondly, on the smart contract side,   though the aggregation logic is simple, it requires many for-loops to read many local models for many aggregation rounds. This computation cost is another concern. Thirdly, if all the local models are stored as variables inside the smart contract, the storage size can grow infeasibly with many clients joining.

To address the above limitations, we propose the following:
\begin{itemize}
\item \textbf{On-chain storage}: we allow the smart contract to be able to store $K < N$ local models in addition to the global model. This $K$ is the value $K$ in the conventional FedAvg training algorithm (Step 2a of Algorithm \ref{alg:FedAvg}). When the smart contract receives a local model, it puts the model in an aggregation queue. When the queue is full with $K$ models, the global update aggregates these models. The queue is emptied for the next batch of $K$ local models. 
\item \textbf{Off-chain global model}: the smart contract saves a copy of the current global model at some downloadable off-chain address. This way, clients have the convenience of downloading the global model off chain, thus avoiding gas fees and smart contract bottleneck. To ensure the correctness of the off-chain copy, a hash value of its is stored in the smart contract. 
\item \textbf{Off-chain local models}: when a client submits its local model to the smart contract, we require that it  save a copy at some downloadable off-chain address. To guarantee immutability, a hash of the model is attached with each client in the smart contract. The off-chain copy is later downloaded by the verifiers for verification scoring. They avoid having to go on-chain to download the model. 
\end{itemize}
Clients are free to choose where off-chain to store their model, but they must provide a downloadable URL. Alternatively, they can utilize a decentralized storage solution such as IPFS.
We leave this choice of storage beyond this paper.

\subsection{Verification and Aggregation}\label{sec:verificationprocedure}
When a verifier is notified of a verification need, it calls a function in the smart contract to get the list of clients to verify. How to determine this list is a logic in the smart contract to be discussed later.   

\subsubsection{Local Verification}
After downloading its list of clients to verify, each verifier classifies these clients in three verification-score groups,  score 0 (malicious), 1 (benign), or 1/2 (unsure), and sends the scores to the smart contract. We will propose a classification method later in the next section.

\subsubsection{Trust Score}
We note that the conventional FL is a repetitive algorithm running in many rounds in a centralized and coordinated way. Because FedBlock is a decentralized setting on the blockchain, we do not have a central coordination anymore. Instead, the clients, the verifiers, and the smart contract are operating asynchronously, whose collaboration is triggered by blockchain events emitted by the smart contract. As such, verifiers may submit verification scores for different clients to the smart contract at different time.

When a client $i$ has a new verification score, say $s_i \in \{0, 1/2, 1\}$, this score is only temporal at the current time, based on its current local model, provided by one verifier. To take behavioral history into account, the smart contract associates with each client $i$ a ``trust score" $S_i \in [0, 1]$, equal to 1 initially, that represents its long-term trustworthiness.  Specifically, upon receipt of a verification score $s_i$ for client $i$, the smart contract updates this client's trust score to
\begin{align}
S_i \leftarrow  \frac{1}{t_i}\bigg((t_i-1)S_i + s_i\bigg),
\label{eq:trustscore}
\end{align}
where $t_i$ counts the times client $i$ receives a verification score. Intuitively, $S_i$ is the average of verification scores to date. 

\subsubsection{Global Aggregation}
Let the current time be $t$ when the smart contract updates the global model by aggregating the local models in the aggregation queue as discussed above. At this time, each client $i$ has a trust score $S_i^{(t)} \in [0, 1]$. The smart contract computes the new global model by averaging the local models weighted by their size and trust score:
\begin{align}
\mathbf{w}^{(t)} = \frac{\sum_{i \in \mathcal{C}^{(t)}} S_i^{(t)} {|\mathcal{D}_i|} \mathbf{w}^{(t)}_i }{ \sum_{i \in \mathcal{C}^{(t)}} {S_i^{(t)} |\mathcal{D}_i|}}.
\label{eq:fedblockaggregate}
\end{align}
Here, $\mathcal{C}^{(t)}$ is the set of local models in the aggregation queue and $\mathbf{w}^{(t)}_i$ is the local model of each client $i$ in this queue.

\subsubsection{Client Set to Verify}
In FedBlock, the smart contract keeps track of a time-varying subset $\mathcal{M}$ of $M< N$ clients that need verification. Note that this $M$ is not necessarily the same as the number $K$ of models in the aggregation queue. Set $\mathcal{M}$ is regenerated from time to time to ensure every client has a chance to be verified. By default, FedBlock selects the clients for $\mathcal{M}$ uniformly at random but we can consider other selection criteria. For example,  clients who have a lower trust score or have been verified less often are stronger candidates for the verification set. 

\subsubsection{Malicious Verifiers}
Verifiers can themselves be malicious, results may be lost in communication, or fewer verifiers involve than needed. As such,  clients may receive bad or no verification scores.  In theory, our framework allows anyone, not necessarily a client, to serve the role of a verifier. For practical purposes, this requires a good  mechanism to incentivize participation from honest verifiers and discourage the bad. 

To increase the quality of verifiers, one way to find good verifiers is among benign clients. These clients are already contributing to the FL task and  might be interested in being verifiers. In this direction, we propose the following client-as-a-verifier approach:  in each round, verifiers can only be accepted among clients with trust score above 1/2. The verifiers, if indeed benign, will do the verification honestly. The malicious  will submit bad scores. In FedBlock, the verification procedure takes place repeatedly round after round. Because the number of benign clients dominates the malicious,  the quality of trust scoring should improve over the time. Hence, the same is expected for the quality of verification.

\section{FedBlock Backdoor Defense}\label{sec:fedblockimplement}
We now propose a concrete backdoor-defense implementation for the proposed framework. This concerns mainly the verification procedure at the verifiers.  We are driven by the following observations:
\begin{itemize}
    \item In a typical DNN model, the ultimate gradient captures rich information about the training objective. This has been justified in the literature \cite{DBLP:conf/ijcnn/NguyenNNWPNN23, 10.1145/3488932.3517395}. Because compromised models share the same training objective, the gradient information can be useful to detect them.
    \item In a FL system, the clients'  data are often non-IID \cite{ZHU2021371} and so their local updates should diverse in early training iterations. Meanwhile, the poisoned models exhibits similar patterns because they are trained with the same backdoor objective. 
\end{itemize}
We elaborate below and then introduce our implementation.

\subsection{Gradient of Ultimate Layer}\label{subsec:ultimatelayer}

\begin{figure}[tb]
    \centering
\includegraphics[width=0.9\columnwidth]{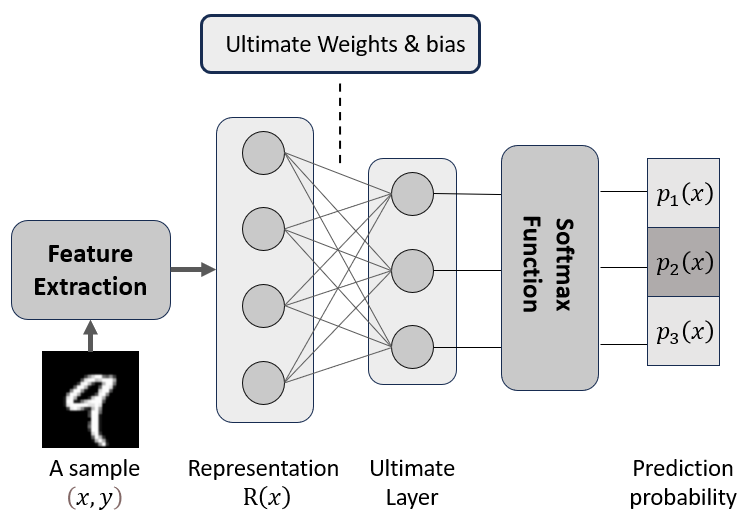}
    \caption{ The ultimate layer is the last layer, which connects from the previous layer via the ultimate weight and bias matrix and to the softmax function.}  
    \label{fig:penultimate}
\end{figure}

As illustrated in Figure \ref{fig:penultimate}, consider a typical deep neural network consisting of a feature extractor and a classifier that is trained using the cross entropy loss by SGD method. The classifier has a dense layer, followed by a softmax activation function, and a bias vector layer. The last layer of neurons is called the ``Ultimate Layer". The number of the ultimate layer's neurons equals to $l$, the number of classes, with the $j^{th}$ neuron indicating the prediction result for the $j^{th}$ class. The weight matrix $\mathbf{U} = [\mathbf{u}_1, ..., \mathbf{u}_l]$ and bias vector $\mathbf{b} = \left < b_1, ..., b_l \right >$ connecting to the ultimate layer from the preceding layer are referred to as ``ultimate weight" and ``ultimate bias". Here $\mathbf{u}_j$ is the $j^{th}$ row of $\mathbf{U}$ and $b_j$ is the $j^{th}$ element of $\mathbf{b}$. 
Intuitively, $\mathbf{u}_j$ and $b_j$ contain the most representative information of the training samples of the $j^{th}$ class.

\begin{definition}[Ultimate Gradient]
The ultimate gradient is the improvement of $\mathbf{U}$ and $\mathbf{b}$   after training the model  with the SGD method. 
Specifically, let $\mathbf{U}^t$ and $\mathbf{b}^t$ be the values of   $\mathbf{U}$ and $\mathbf{b}$ at the $t^{th}$ training iteration, then the ultimate gradient at the $t^{th}$ iteration is the following combination:
\[
\mathbf{g}^t = \bigg(\nabla \mathbf{U}^t=\frac{\mathbf{U}^t - \mathbf{U}^{t-1}}{-\eta},  \nabla \mathbf{b}^t=\frac{\mathbf{b}^t - \mathbf{b}^{t-1}}{-\eta} \bigg),
\]
where $\eta$ is the learning rate. $\nabla \mathbf{U}^t$ is called the ultimate weight gradient and $\nabla \mathbf{b}^t$ is called the ultimate bias gradient. 
\end{definition}
To ease the presentation, we may omit the index $t$ and simply represent the ultimate gradients by $\mathbf{g}$, $\mathbf{U}$, and $\mathbf{b}$.

\begin{definition}[By-Class Ultimate Gradient]
This quantity highlights the effect of the training data on the ultimate layer. It is  obtained by summing the ultimate gradient by class
 \[
 \boldsymbol{\mu} = \left<\mathbf{g} \cdot \mathbf{1}\right> {=} \left<\nabla \mathbf{U} \cdot \mathbf{1}, \nabla \mathbf{b}\right>.
 \]
 Note that $\boldsymbol{\mu}$ is a vector concatenating two $l$-dimensional vectors: the by-class sum of the ultimate weight gradient vector $\nabla \mathbf{U}$ and the ultimate bias gradient vector  $\nabla \mathbf{b}$. 
\end{definition}

\begin{remark}[1] If we train on two datasets with same distribution, the resulting ultimate weight gradient ($\nabla \mathbf{U}$) should share the same pattern. Because backdoor attackers need a sufficient poisoned data rate to maximize the impact and apply the same attack objective, the compromised models should have similar distributions, thus similar ultimate gradients. On the other hand, the training data of benign FL clients are non-iid in practice. Their ultimate gradients should exhibit different patterns in the early FL rounds. Hence, the attack could be detected by looking at the similarity of ultimate gradients. 
\end{remark}

\begin{remark}[2] The above remark is less useful in later FL rounds. However, in backdoor attacks, the by-class ultimate gradients ($\boldsymbol{\mu}$) of the compromised clients' models tend to have similar patterns. Also, these patterns differ from that of the benign. 
Hence,  quantity $\boldsymbol{\mu}$ is useful for us to distinguish them. Note that similar by-class gradients do not necessarily imply similar gradients.
\end{remark}

\subsection{Ultimate Gradient based Verification}\label{subsection:verify}
The above insights are extremely useful. Instead of processing the raw training models which are large in size and may not offer significant discriminating hints, we only need to look at the ultimate gradient information to detect malicious clients. 

Recall that $\mathbf{w}_i^{(t)}$ denotes client $i$'s raw local model at current time $t$. Let $\mathbf{U}_i^{(t)}$ and $\mathbf{b}_i^{(t)}$  denote its ultimate weight and ultimate bias, respectively. The corresponding  ultimate gradients are $\nabla \mathbf{U}_i^{(t)}$ and $\nabla \mathbf{b}_i^{(t)}$. For ease of presentation,  superscript $t$ is omitted when the context is understood.

Consider a particular verifier and let $\mathcal{C}$ be its client verficication set.
In our proposed FedBlock implementation, the verifier will download the ultimate gradient data, not the raw model data; specifically, 
\[
\bigg \{  (\nabla\mathbf{U}_i, \nabla\mathbf{b}_i ) ~|~ i \in \mathcal{C} \bigg\}.
\]
It then will assign a score to each client in $\mathcal{C}$ according to a verification procedure as follows.

\subsubsection{Gradient Similarity Filter}
According to Remark (1), the gradients of compromised clients should be similar in pattern, and on the other hand, that of the benign should diverse in early rounds. We apply the following filter.
Let
\[
\mathbf{U}_{*} = 
\frac{\sum_{i \in \mathcal{C}} | \mathcal{D}_i |  \mathbf{U}_i}{\sum_{i \in \mathcal{C}} | \mathcal{D}_i |}, 
\nabla \mathbf{U}_{*} = 
\frac{\sum_{i \in \mathcal{C}} | \mathcal{D}_i | \nabla \mathbf{U}_i}{\sum_{i \in \mathcal{C}} | \mathcal{D}_i |}
\]
be the average ultimate weight and average ultimate weight gradient over clients in $\mathcal{C}$.
Then, we compute a score for each client $i$ 
\[
a_i = \frac{\mathbf{U}_i- \mathbf{U}_{*} }{\| \mathbf{U}_i-\mathbf{U}_*\|} \cdot \frac{\nabla \mathbf{U}_{*}}{\|\nabla \mathbf{U}_{*}\|}.
\]
and normalize it  using min-max scaling 
\[
a_i  \leftarrow \frac{a_i - \min_j\{a_j\}}{\max_j\{a_j\} - \min_j\{a_j\}}.
\]
It is intuitive that this score tends to be higher for compromised models than that of the benign. We then use a threshold 
\[
S_1 = \bigg\{ i  \in \mathcal{C} ~|~  a_i >  \underset{j} {\text{median}} ~ \{a_j\} \bigg\},
\]
to categorize clients in $S_1$ as   malicious. 

\subsubsection{By-Class Gradient Similarity Filter}
Although diverse in early rounds, the local models of benign clients should behave similarly in  late rounds of FL global update because they are aimed to converge globally. As such, the above filter can falsely classify benign clients as malicious. 

According to Remark (2), the by-class ultimate gradients of malicious models tend to have similar patterns, that of the benign in late rounds  are also similar, but the malicious pattern and the benign pattern are different.  We thus apply another filter based on k-mean clustering to distinguish these two groups. Here, each client $i$ is a  ``point" to put in k-mean whose feature vector consists of L2 distances of its by-class ultimate gradient to that of the other clients:
$ \mathbf{x}_i = \bigg < \| \boldsymbol{\mu}_i - \boldsymbol{\mu}_1\|, \| \boldsymbol{\mu}_i - \boldsymbol{\mu}_2\|, \dots, \| \boldsymbol{\mu}_i - \boldsymbol{\mu}_{|\mathcal{C}|}\| \bigg >$. 
To find which cluster is the malicious, we rely on trust scores. As we discuss in Section \ref{sec:verificationprocedure}, each client has a long-term trust score. We decide that the malicious cluster is the one containing the lowest-score client. Let $S_2$ be the set of clients categorized as malicious.

\subsubsection{Combining Filters}
After obtaining suspected sets $S_1$ and $S_2$, the verifier assigns client $i \in \mathcal{C}$ the following score
\[
s_i = 
\left\{
\begin{array}{lr}
    0 & \text{if } i \in S_1 \cap S_2\\
    1 & \text{if } i \not \in (S_1 \cup S_2)\\
    1/2 & \text{otherwise.}
\end{array}
\right.
\]
The verifier then submits the verification scores of all clients in $\mathcal{C}$ to the smart contract. The smart contract uses these scores to update their long-term trust score according to Eq. \eqref{eq:trustscore}. It then updates the global model by aggregating the local models weighted by these trust scores according to Eq. \eqref{eq:fedblockaggregate}. 

\subsection{Quality Control of Verification}
To increase the quality of the verification, it is important to use as many verifiers and provide scores for as many clients as possible. How many is enough? In practice, we need to incentivize verifiers. More verifiers means a higher incentive budget, not to mention the case we cannot find that many available. On the verifier side, more clients to verify means more communication and computation costs for the verifiers. We need a good balance.

Consider the verification of client set $\mathcal{M}$ of size $M$ (currently selected by the smart contract).  Assume $V$ verifiers, each verifying $L$ random clients from $\mathcal{M}$. We want to find 
the minimum $V$ and $L$ such that each client in $\mathcal{M}$ is verified by at least one verifier. We have the following important results.

\begin{theorem}[Optimal value for $L$]
Given $V$ verifiers and $M$ clients, the expected value of $L$ clients that each verifier needs to verify such that all $M$ clients are verified is 
\begin{align}
\mathbb{E}[L]= \sum_{l=1}^M\sum_{s=1}^M \binom{M}{s} (-1)^{s+1} 
\left[ \binom{M-s}{l} / \binom{M}{l} \right]^V.
\label{eq:L}
\end{align}
\label{theorem:L}
\end{theorem}
 \begin{proof}
     Consider a subset ${S}  \subset \mathcal{M}$. If the  need is to verify $l$ clients, the probability of choosing none from set $\mathcal{S}$ is
\begin{equation}
Q({S})=\left[\binom{M-|{S}|}{l} / \binom{M}{l} \right]^V.
\end{equation}
The probability that when $L > l$ all $M$ clients are verified is:
\begin{align}
\mathrm{P}(L > l)= 
\sum_{S \subset \mathcal{M}, |S|=1} Q(S)-\sum_{S \subset \mathcal{M} ,|S|=2} Q(S)\\
+\sum_{S \subset \mathcal{M}, |S|=3} Q(S)- \sum_{S \subset \mathcal{M}, |S|=4} Q(S) + \ldots \\
=\sum_{s=1}^n \binom{n}{s} \cdot(-1)^{s+1}\cdot
\left[\binom{n-s}{l} / \binom{n}{l} \right]^V.
\end{align}
So, the expected  $L$ such that all $M$ clients are verified is: 
\begin{align*}
\mathbb{E}[L]=\sum_{l=1}^n P(L > l)\\
=\sum_{l=1}^n\sum_{s=1}^n \binom{n}{s} (-1)^{s+1} \left[
\binom{n-s}{l} / \binom{n}{l} \right]^V.
\end{align*}
\end{proof}

\begin{theorem}[Optimal value for $V$]
Given $M$ and $L$, the expected value of $V$ of verifiers, each verifying $L$ clients, such that all $M$ clients are verified is
\begin{align}
\mathbb{E}[V]
= \binom{M}{L}  \sum_{s=1}^M 
\dfrac{(-1)^{s+1} \binom{M}{s}}{ \binom{M}{L} -\binom{M-s}{L}}.
\label{eq:V}
\end{align}
\label{theorem:V}
\end{theorem}

\begin{proof}
    Consider a subset ${S} \subset \mathcal{M}$. If we have $v$ verifiers, the probability that all $v$ verifiers do not choose any client in $\mathcal{S}$ is 
\begin{align}
P(S)=\left[{\binom{M-|\mathcal{S}|}{L}}/{\binom{M}{L}}\right]^v.
\end{align}
The probability that more than $v$ verifiers ($V > v$) select all $M$ clients is equal to the probability that exactly $v$ verifiers do not select all $K$ clients, and it is equal to:
\begin{align}
\mathrm{P}(V > v)= \sum_{S \subset \mathcal{M}, |S|=1} P(S)-\sum_{S \subset \mathcal{M} ,|S|=2} P(S)\\
+\sum_{S \subset \mathcal{M}, |S|=3} P(S)- \sum_{S \subset \mathcal{M}, |S|=4} P(S) + \ldots \\
=\sum_{s=1}^M \binom{M}{i} (-1)^{s+1}\left[\binom{M-s}{L}/\binom{M}{L}\right]^v.
\end{align}
So, the expected $V$ such that all $M$ clients are verified is:
\begin{align*}
\mathbb{E}[V]=\sum_{v=0}^{\infty} P(V > v)
= \binom{M}{L}  \sum_{s=1}^M 
\dfrac{(-1)^{s+1} \binom{M}{s}}{ \binom{M}{L} -\binom{M-s}{L}}.
\end{align*}
\end{proof}

The above theorems provide theoretical guidelines we should follow to ensure verification effectiveness. 
Whether we have the expected number of verifiers to join and whether they are cooperative  affect the quality of verification. Using our proposed client-as-a-verifier approach can help. Regardless, there must be incentives for honest verifiers. This is common in blockchain systems that require collaborators. FedBlock is a framework and we leave this incentive issue to actual real-world implementations, hence beyond the scope of this paper.

\section{Evaluation}\label{sec:eval}
We conducted an experiment to evaluate FedBlock  against backdoor attacks and also compared it to a recent representative of  FL backdoor defense which does not adopt blockchain. 

\subsection {Experimental Setup}

\textbf{FL algorithm setting.} We use the model averaging and random client selection approach of the FedAvg algorithm for each  aggregation round in FedBlock. The training and test data are from two datasets, EMNIST using the LeNet model  \cite{lecun} and CIFAR-10 using the Zisserman model \cite{DBLP:journals/corr/SimonyanZ14a}. 
The number of FL clients is $N=3383$ for  EMNIST   and $N=200$ for   CIFAR-10. In each  round (500 by default),  $K=30$ random clients are selected to receive the global update. The client data distribution is non-iid. These datasests are benchmark for evaluation in FL literature and these parameters also used in earlier works on FL+Backdoor \cite{DBLP:conf/ijcnn/NguyenNNWPNN23,wang2020attack}. 

\textbf{Adversary setting.} We consider rigorous backdoor attack scenarios by adjusting three parameters: 1) attacker ratio $(\epsilon)$: the portion of compromised clients out of all clients, 2) poisoned data rate ($PDR$): the portion of poison inserted in client data;  and 3) non-IID degree $(\varphi)$: the heterogeneity of training data distribution over clients.  Unless otherwise mentioned, by default,  $\epsilon = 25\%$, $\varphi=0.5$, and $PDR=33\%~ (50\%)$  for CIFAR-10 (EMNIST). These and other unmentioned parameters for reproducing are inherited from  \cite{wang2020attack, xie2021crfl}.

\textbf{Backdoor attack strategies.} We evaluated FedBlock  under  edge-case based attacks, the most difficult type of backdoor. Specifically, we implemented 
three edge-case strategies proposed in \cite{wang2020attack}: black-box attack, projected gradient descent (PGD) attack, and PGD combined  with model replacement (PGD+MR). In PGD, the projected frequency is set to 1 since there may be malicious clients in each aggregation round.

\textbf{Verification parameters.}
In the smart contract, we set $M=30$, i.e., in each aggregation round the smart contract triggers a need for the verifiers to score collectively 30 clients. We evaluated with various settings but, by default, there are $V=15$ verifiers, each expected to score $L=7$ clients among these 30 clients. This choice is suggested by Theorems \ref{theorem:L} and \ref{theorem:V}. 


\begin{figure*}[t]
    \begin{minipage}{0.43\linewidth}
    \centering
    \captionof{table}{FedBlock vs. FedGrad under different backdoor strategies (CIFAR-10 and EMNIST,  500  rounds).}
    \label{table:cen-decen}
    \setlength\tabcolsep{3pt} 
    \resizebox{\linewidth}{!}{%
        \begin{tabular}{c|c|c|c|c|c|c|c|c|c}
            \toprule
            \multicolumn{2}{c|}{\multirow{2}{*}{\textbf{Attack type}}} & \multicolumn{4}{c|}{CIFAR-10} & \multicolumn{4}{c}{EMNIST} \\
            \cmidrule(lr){3-6} \cmidrule(lr){7-10}
            \multicolumn{2}{c|}{} & {MA} & {BA} & {TPR} & {TNR} & {MA} & {BA} & {TPR} & {TNR}\\
            \midrule
            \multirow{2}{*}{\begin{tabular}[c]{@{}c@{}}Black\\-box\end{tabular}} 
            & FedGrad &  84.1 & 1.5 & 1.0 & 0.67 
            & 98.28 & 7.0 & 1.0 & 0.68\\
            & \textbf{FedBlock} & \textbf{84.2} & \textbf{1.0} & 1.0 & \textbf{0.74}
            & \textbf{98.65} & 7.0 & 1.0 & \textbf{0.83}\\
            \midrule
             \multirow{2}{*}{PGD} & FedGrad & 83.8 & \textbf{1.0} & 1.0 & 0.67
            & 98.36 & \textbf{7.0} & 1.0 & 0.67 \\
            & \textbf{FedBlock} & \textbf{84.0} & 2.3 & 1.0& \textbf{0.74}
            & \textbf{98.61} & 8.0 & 1.0 & \textbf{0.83} \\
            \midrule
            \multirow{2}{*}{
            \begin{tabular}[c]{@{}c@{}}PGD\\+MR\end{tabular}} 
            & FedGrad & 84.6 & \textbf{2.8} & 1.0 & 0.67
            & 98.43 & 8.0 & 1.0 & 0.68 \\
            & \textbf{FedBlock} & \textbf{84.6} & 3.6 & 1.0 & \textbf{0.81}
            & \textbf{98.62} & \textbf{6.0} & 1.0 & \textbf{0.84} \\
            \bottomrule
        \end{tabular}
    }

    \end{minipage} 
    \hfill
    \begin{minipage}{0.55\linewidth}
        \centering
        \includegraphics[width=1\linewidth]{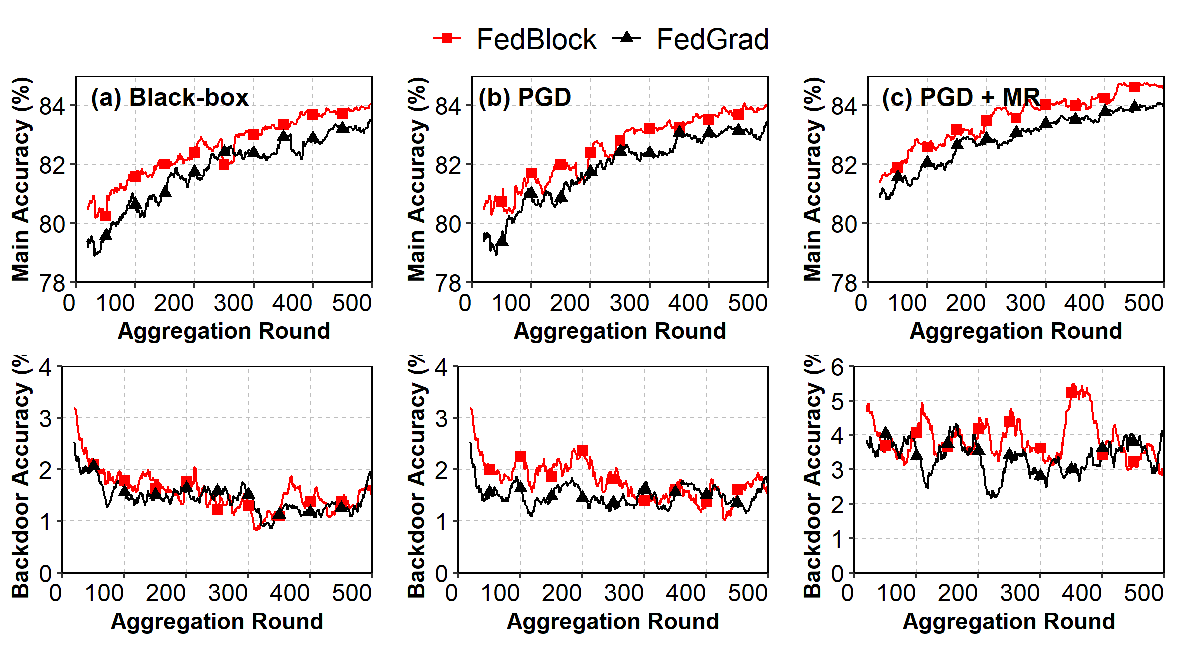}
        \caption{Real-time performance in individual aggregation rounds (CIFAR-10). For better visualization, here we plot the results that are smooth-averaged after every 20 rounds.}
        \label{fig:robust-real-time}
    \end{minipage}
    \vspace{-0.3cm}
\end{figure*}

\textbf{Evaluation metrics.} We used Main-task Accuracy (MA) and Backdoor Accuracy (BA), two usual metrics to measure the impact on learning accuracy in defense against backdoors  \cite{bagdasaryan2020backdoor, wang2020attack}. MA represents the prediction accuracy of the main task. BA represents the prediction accuracy of the attacker's targeted learning goal. We also computed two other metrics, True Positive Rate (TPR) and True Negative Rate (TNR), to measure the accuracy in identifying truly malicious clients and truly benign clients, respectively. A good defense should offer high MA, low BA, high TPR, and high TNR.

\subsection{Comparison to the centralized non-blockchain benchmark}

We now discuss the results of comparing FedBlock to FedGrad \cite{DBLP:conf/ijcnn/NguyenNNWPNN23}. We chose FedGrad for benchmarking because i) it is a recent backdoor defense proposal for conventional non-blockchain FL, ii)
it had been shown to compete well with the earlier literature, and iii) our implementation of FedBlock is similar to FedGrad in the use of ultimate layer's model information to filter out malicious models. Our goal was to investigate how  FedBlock could approach or even outperform   FedGrad. 
The comparison was conducted using three types of backdoor attack: Blackbox, PGD, and PGD+MR. The parameter configurations, as mentioned in the previous section, were set the same for both FedBlock and FedGrad. 

\textbf{Accuracy comparison. }
Table \ref{table:cen-decen} shows the results. Both methods are observed strongly capable to identify malicious models, as evidenced by their TPR reaching 1.0. This is further supported by seeing resonably high values for MA and reasonably low values for BA. However, noticeably,  FedBlock offers a much higher TNR than FedGrad in all comparison scenarios. In other words, FedBlock is more effective in detecting benign clients.  FedBlock   also outperforms FedGrad in MA for all scenarios, meaning better  prediction for the main learning task. In terms of BA, FedBlock and FedGrad are  comparable, having very low BA, insignificant. The above observations apply to both CIFAR-10 and EMNIST datasets.

A zoom-in for each aggregation round is seen in Figure \ref{fig:robust-real-time} for CIFAR-10 as an example. They both are observed to improve  over the time. In early rounds, MA is low and BA is high. This is understandable. For every FL system, the training can only converge to a good accuracy after a sufficient number of aggregation rounds. This explains the low values for MA. Early rounds are the most vulnerable as backdoor attacks can easily make effect, hence high values for BA. With more rounds, the defense starts to see its impact. It is consistently observed that FedBlock has better MA. FedGrad offers better BA but only slightly, and BA in both methods is very low.

\textbf{Time comparison. }
An obvious disadvantage of  centralized FL backdoor-defense techniques, represented by FedGrad, is the computation time because the server has to execute the detection algorithm over the entire set of $M$ clients. FedBlock, by decentralizing this task to $V$ verifiers each computing only a subset of $L$ clients, the computation overhead per verifier is  less. Since this computation by each verifier is done in parallel, the entire verification can be completed faster, hence speeding up each aggregation round. To validate this point, we have Figure \ref{figure:cen-decen} plotting the average time to complete each aggregation round in FedBlock and FedGrad. FedBlock is at least three times faster. This is a great advantage of a decentralized solution and FedBlock owns that.

\subsection{Effect of  different backdoor intensity and  non-iid-ness}

\begin{figure*}[tb]
    \centering
    \begin{minipage}{0.3\linewidth}
        \centering
        \includegraphics[width=0.95\linewidth]{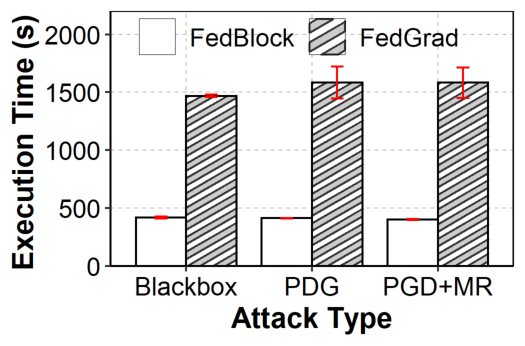}
        \caption{FedBlock vs. FedGrad: average execution time per aggregation round (CIFAR-10, 500 rounds).} 
        \label{figure:cen-decen}
    \end{minipage}
    \hfill
    \begin{minipage}{0.65\linewidth}
    \begin{center}
    \centering
    \includegraphics[width=0.6\linewidth]{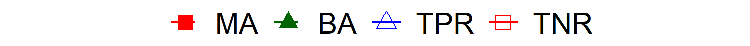}\\
     {\includegraphics[width=0.325\linewidth]{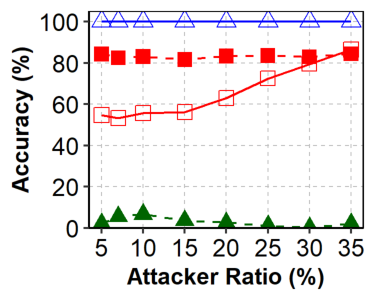}} 
     {\includegraphics[width=0.325\linewidth]{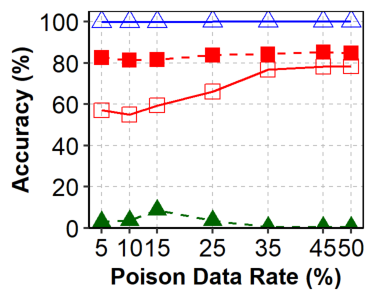}} 
     {\includegraphics[width=0.325\linewidth]{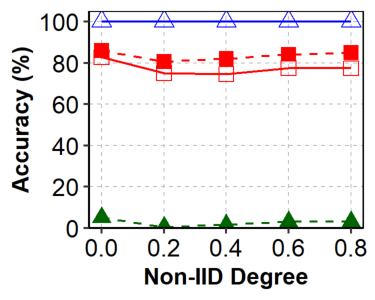}}    
        \caption{FedBlock under variations in (1) attacker ratio, (2) poisoned data rate, and (3) non-IID degree. Here, the results are for CIFAR-10 dataset, black-box attack, and 500 rounds.}
        \label{fig:robust}     
    \end{center}
 
    \end{minipage}
\end{figure*}

\begin{figure*}[tb]
    \begin{minipage}{0.45\linewidth}
        \centering
    \centering
    \captionof{table}{FedBlock under the effect of bad verifiers and effectiveness of the client-as-a-verifier (CAAV) approach (500 rounds, CIFAR-10, black-box backdoor).}
    \label{table:verify_attacker_free}
    \setlength\tabcolsep{3pt} 
    \resizebox{0.9\linewidth}{!}{%
        \begin{tabular}{c|c|c|c|c|c}
            \toprule
            \multicolumn{1}{c|}{\textbf{Attack type}} & \textbf{\% bad verifiers} & {MA} & {BA} & {TPR} & {TNR} \\
            \midrule
            Normal & $0  \%$ & {84.20} & {1.00} & 1.0 & {0.74}\\ \midrule
            \multirow{3}{*}{\begin{tabular}[c]{@{}c@{}}Random\end{tabular}}
            & 10  \%& 83.16 & 17.85 & 0.99 & 0.67 \\
            & 20  \%& 82.18 & 67.34 & 0.97 & 0.57 \\
            & 30  \%& 82.61 & 64.28 & 0.96 & 0.52 \\ 
             & \textbf{CAAV}  & \textbf{83.59} & \textbf{{0.51}} & \textbf{1.0} & \textbf{0.74} \\ \midrule
            \multirow{3}{*}{\begin{tabular}[c]{@{}c@{}}Reverse\end{tabular}}
            & 10  \%& 83.86 & 58.16 & 0.98 & 0.62 \\
            & 20  \%& 84.63 & 72.96 & 0.93 & 0.50 \\
            & 30  \%& 83.80 & 76.53 & 0.54 & 0.36 \\ 
            & \textbf{CAAV}   & \textbf{84.71} & \textbf{{0.51}} & \textbf{1.0} & \textbf{0.75}\\ 
            \bottomrule
        \end{tabular}
    }
    \end{minipage}
    \hfill
    \begin{minipage} {0.53\linewidth}
        \centering
        \captionof{table}{FedBlock using client-as-a-verifier (CAAV) under different verifier-attack and backdoor-attack strategies (500 rounds, CIFAR-10 and EMNIST).}
        \label{table:verify_attacker}
        \setlength\tabcolsep{3pt} 
        \resizebox{1.0\linewidth}{!}{%
            \begin{tabular}{c|c|c|c|c|c|c|c|c|c}
                \toprule
                \multicolumn{2}{c|}{\multirow{2}{*}{\textbf{Attack type}}} & \multicolumn{4}{c|}{CIFAR-10} & \multicolumn{4}{c}{EMNIST} \\
                \cmidrule(lr){3-6} \cmidrule(lr){7-10}
                \multicolumn{2}{c|}{} & {MA} & {BA} & {TPR} & {TNR} & {MA} & {BA} & {TPR} & {TNR}\\
                \midrule
                \multirow{3}{*}{\begin{tabular}[c]{@{}c@{}}Black\\-box\end{tabular}} 
                & Normal& {84.20} & {1.00} & 1.0 & {0.74}
            & {98.65} & 7.0 & 1.0 & {0.83} \\
                & CAAV: Random  &  83.59 & {0.51} & 1.0 & 0.74 
                & 98.63 & {3.0} & 1.0 & 0.82\\
                & CAAV: Reverse  & 84.71 & {0.51} & 1.0 & 0.75
                & 98.62 & 6.0 & 1.0 & 0.83\\
                \midrule
                 \multirow{3}{*}{PGD} 
                 & Normal  & {84.00} & 2.30 & 1.0& {0.74}
            & {98.61} & 8.0 & 1.0 & {0.83} \\
                 & CAAV: Random &  83.62 & {0.51} & 1.0 & 0.73 & 98.63 & {3.0} & 1.0 & 0.82\\
                 & CAAV: Reverse  &  84.06 & 3.06 & 1.0 & 0.69 
                & 98.63 & 5.0 & 1.0 & 0.82\\
                \midrule
                \multirow{3}{*}{
                \begin{tabular}[c]{@{}c@{}}PGD\\+MR\end{tabular}} 
                & Normal&  {84.60} & 3.60 & 1.0 & {0.81}
            & {98.62} & {6.0} & 1.0 & {0.84} \\
                & CAAV: Random  &  85.31 & 4.08 & 1.0 & 0.81 
                & 98.62 & {2.0} & 1.0 & 0.83\\
                & CAAV: Reverse  &  83.71 & 6.63 & 1.0 & 0.82 
                & 98.60 & 6.0 & 1.0 & 0.84\\
                \bottomrule
            \end{tabular}
        }
    \end{minipage} 
\end{figure*}

We investigated FedBlock stability under increasing intensity of backdoor attack by varying these parameters: 1) attacker ratio $\epsilon \in [5\%, 35\%]$: the portion of compromised clients out of all clients, 2) poisoned data rate $PDR \in [5\%, 50\%$]: the portion ratio of poison inserted in client data;  and 3) non-IID degree $\varphi \in [0, 0.8]$.  Figure~\ref{fig:robust} provides  a representative example for discussion, showing the results for the CIFAR-10 dataset, black-box attack, and 500 aggregation rounds. 
It is clearly seen that FedBlock is  stable in all metrics. In all scenarios, FedBlock offers TPR 1.0, always excellent in removing suspicious attackers. Also, MA remains high and BA low, indicating FedBlock's sustainable high main-task accuracy and low backdoor-targeted accuracy.

Noticeably observed, increasing backdoor intensity beyond a certain threshold does not  benefit the attacker. This is evidenced with lower BA and higher TNR seen in Figure \ref{fig:robust} (1, 2). Too much poison, by increasing the attacker ratio or PDR excessively, actually harms their success rate in targeted goal (lower BA)  and makes FedBlock easier to detect benign clients (higher TNR). This result shows an interesting strength of FedBlock against backdoors committed by  greedy attackers.

Also, Figure \ref{fig:robust} (3)    shows that, when the training data distribution over the clients is more non-iid enough (increasing $\varphi$), FedBlock increases its capability to identify benign models (higher TNR value). This is understandable because, as explained in Remark (1) of Section \ref{subsec:ultimatelayer}, with non-iid distribution, the benign's local models are more diverse in early rounds, making FedBlock more effective in distinguish them from malicious models. In practice, FL accuracy suffers from increasing non-iid-ness, but here we present an interesting finding that this non-idd-ness can be a good discriminator against backdoors and it is utilized in FedBlock.

\subsection{Effect of Malicious Verifiers}

Next, we  evaluated the case where, besides backdoor attacks affecting client models, a subset of verifiers  are also dishonest by submitting bad scores to the smart contract.  We considered two types of verifier-attack: i) Random score: randomize scores to submit to the smart contract; and ii) Reverse score: reverse scores such that score 0 becomes 1, score 1 becomes 0, and score 1/2 unchanged.

For the example case of CIFAR-10 and block-box backdoor, Table \ref{table:verify_attacker_free} shows the performance of FedBlock where verifiers can be anyone, not necessarily clients, when $p = 10\%, 20\%, 30\%$ of them are dishonest.  It is no surprise that the presence of bad verifiers leads to degrading performance. However,  consider our proposed client-as-a-verifier (CAAV) approach, which considers only clients with trust score above 1/2 as verifiers. For both kinds of verifier-attack, random score or reverse score, CAAV is not impacted at all. As we explained in the previous section, this is because our trust scoring quality improves over the time  and so CAAV results in more and more good verifiers being selected in the verification procedure. 

Table \ref{table:verify_attacker} provides more experimental evidence to substantiate the robustness of CAAV under verifier attacks, which applies to all experimented cases of datasets, back-door attacks, and verifier-attacks. Here we see some slight differences in BA, but all these BA values are small, considered good.

\section{Conclusions}\label{sec:conclusion}
We have proposed FedBlock, a blockchain-based FL framework robust to  backdoor attacks. Despite growing research in FL backdoor defense, the server has always been assumed honest. Recently, blockchain has been utilized to decentralize the server role. However,  such a setting does not touch on backdoor security, except rare work which, however, requires building a new blockchain from scratch. In contrast, FedBlock can run on any blockchain network. The proposed framework introduces the role of verifiers who in a decentralized manner detect malicious local model updates so that they are removed from global aggregation. Aware of blockchain's gas-fee constraints, FedBlock does not require on-chain storage of all client models. Instead, it stores only a small amount on-chain while the rest is off-chain. Compared to the literature of FL centralized backdoor defense, our evaluation has shown that FedBlock is highly competitive in the most important accuracy metrics while offering faster FL training time. Specially, server attack is non-existent in FedBlock thanks to blockchain. 

\textbf{Limitations and future work}. FedBlock is a fresh direction with plenty of room to improve and explore. There are issues we have not touched deeply. Firstly, with possibility of malicious verifiers, our proposed client-as-a-verifier approach can be highly effective, but if anyone can be a verifier, we need more than just assuming that the majority of verifiers is honest and will prevail over many FL rounds. We need a mechanism to detect and ignore malicious verifiers. Secondly, for FedBlock to practically be usable, we need an incentive mechanism to encourage good participation from verifiers. Last but not least, in the current version, the smart contract selects random clients to verify in each round, but this selection criterion might not be optimal. In the future work, we will investigate the aforementioned issues as well as build a real-blockchain prototype for FedBlock.

\section*{Acknowledgements}
This work was supported by Vingroup Innovation Foundation (VINIF) under project code VINIF.2021.DA00128.

\bibliographystyle{IEEEtran}
\bibliography{bibliography, misc}

\end{document}